\documentclass[journal]{IEEEtran}
\usepackage{amsmath,amsfonts}
\usepackage{amsthm}
\usepackage{mathtools}
\usepackage{algorithm}
\usepackage{algpseudocode}
\usepackage{array}
\usepackage[caption=false,font=normalsize,labelfont=sf,textfont=sf]{subfig}
\usepackage{textcomp}
\usepackage{stfloats}
\usepackage{url}
\usepackage{verbatim}
\usepackage{graphicx}
\usepackage{makecell}
\usepackage{multirow}
\usepackage{multicol}
\usepackage{bm}

\usepackage{cite}

\def\BibTeX{{\rm B\kern-.05em{\sc i\kern-.025em b}\kern-.08em
    T\kern-.1667em\lower.7ex\hbox{E}\kern-.125emX}}
\usepackage{balance}

\newtheorem{definition}{Definition}
\newtheorem{theorem}{Theorem}

\newcommand{\Rmnum}[1]{\uppercase\expandafter{\romannumeral #1}}
\begin{document}
\title{Large-Scale FPGA-Based Privacy Amplification Exceeding $10^8$ Bits for Quantum Key Distribution}

\author{
    Xi Cheng, Hao-kun Mao, Hong-wei Xu, and Qiong Li$^{*}$%

    \thanks{Xi Cheng, Hao-Kun Mao, Hong-wei Xu, and Qiong Li and are with the Faculty of Computing, Harbin Institute of Technology, Harbin, Heilongjiang, 150000, China}

    \thanks{$^{*}$Qiong Li is the corresponding author. (e-mail: qiongli@hit.edu.cn).} 
    \thanks{This work is supported by National Natural Science Foundation of China (Grant number: 62071151), and Innovation Program for Quantum Science and Technology (Grant No. 2021ZD0300701).}
}

\markboth{Journal of \LaTeX\ Class Files,~Vol.~18, No.~9, September~2020}%
{How to Use the IEEEtran \LaTeX \ Templates}

\maketitle
\begin{abstract}

Privacy Amplification (PA) is indispensable in Quantum Key Distribution (QKD) post-processing, as it eliminates information leakage to eavesdroppers. Field-programmable gate arrays (FPGAs) are highly attractive for QKD systems due to their flexibility and high integration. However, due to limited resources, input and output sizes remain the primary bottleneck in FPGA-based PA schemes for Discrete Variable (DV)-QKD systems. In this paper, we present a large-scale FPGA-based PA scheme that supports both input block sizes and output key sizes exceeding $10^8$ bits, effectively addressing the challenges posed by the finite-size effect. To accommodate the large input and output sizes, we propose a novel PA algorithm and prove its security. We implement and evaluate this scheme on a Xilinx XCKU095 FPGA platform. Experimental results demonstrate that our PA implementation can handle an input block size of $10^8$ bits with flexible output sizes up to the input size. For DV-QKD systems, our PA scheme supports an input block size nearly two orders of magnitude larger than current FPGA-based PA schemes, significantly mitigating the impact of the finite-size effect on the final secure key rate.

\end{abstract}

\section{Introduction}
\IEEEPARstart{Q}{uantum} Key Distribution (QKD) is a secure communication technology based on quantum mechanics, enabling two legitimate parties, named Alice and Bob, to generate a shared key that is information-theoretically secure. \cite{Bennett1984}. In recent years, the performance and practicality of QKD systems have been significantly enhanced, with transmission distances surpassing 1000 kilometers \cite{liu2023experimental} and the highest secure key rate reaching 110 Mbps \cite{Li2023}. QKD is increasingly demonstrating its pivotal role in safeguarding communications against the threat of eavesdroppers.

Practical QKD systems consist of two main components: the quantum communication subsystem and the post-processing subsystem \cite{Fung2010}. The post-processing subsystem primarily involves two critical steps: information reconciliation, and privacy amplification. As the final step of the QKD system, privacy amplification is particularly vital as it transforms error-corrected keys that are partially secure into highly secure keys, thus ensuring robust security in QKD \cite{Bennett1988, Bennett1995}.

For the DV-QKD system, the privacy amplification has to support a large input block size and an output size of comparable magnitude.\cite{Li2023}. The input block size of the PA scheme is critical to the QKD secure key rate. Research on the asymmetric BB84 protocol indicates that an input block size of $10^6$ bits can mitigate the finite-size effect \cite{Tomamichel2012}, however, this block size is not sufficient to achieve a high secure key rate. $10^8$ bits is considered appropriate to efficiently mitigate the finite-size effect \cite{Yuan201810}, which allows 85\% of the asymptotic secure key rate in T12 protocol \cite{Lucamarini2013}. 
In practical high rate DV-QKD systems \cite{Yuan201810, Li2023}, the maximum output size is commonly limited to 1/3 or 1/2 of the input block size. Therefore, PA schemes must support output sizes comparable in scale to the input block size; otherwise, a significant portion of the secure key may be lost.

Among current hardware platforms, FPGAs are widely used in QKD systems\cite{Yuan201810}, \cite{Zhang2012,Tanaka2012,walenta2014fast,lu2015fpga,Constantin2017,stanco2022versatile}, to control the optical components, distribute the synchronization signals, and accelerate post-processing operations \cite{Yang2020, Venkatachalam2023scalable}.
With the iterative updates of PA schemes, FPGAs offer flexible programming capabilities, making them a versatile and cost-effective solution. Moreover, FPGAs have proven to be effective tools for rapid prototyping prior to chip integration.

However, achieving both large input block sizes and large output sizes remains a significant challenge for existing FPGA-based PA schemes \cite{Yan2023}.
Most FPGA-based PA schemes for DV-QKD \cite{Zhang2012, Tanaka2012, Constantin2017, Yang2017, Li2019} utilize the Toeplitz matrix for secure key distillation. This hash family is particularly well-suited for FPGA implementation due to its simple construction, which allows for flexible partitioning and parallel processing.
These PA schemes can support output sizes comparable to the input block size. However, due to the limitation in FPGA resources, these FPGA-based PA schemes are typically designed to support input block sizes on the order of $10^6$ bits. 
In 2022, a high throughput FPGA-based PA scheme was proposed, achieving a substantial input block size of $10^8$ bits \cite{Yan2023}. However, when the block size reaches $10^8$ bits, the maximum output size of this PA scheme falls below $8 \times 10^5$ bits.

In this paper, we propose a large-scale FPGA-based PA scheme. It achieves an input block size of $10^8$ bits while maintaining an output size of comparable scale. To support the large input and output sizes, we introduce a new hash family based on Multilinear-Modular Hashing (MMH) and prove its universal properties. Leveraging the new universal hash family, we proposed a novel PA algorithm that supports adaptive input and output sizes. Experimental results demonstrate that this scheme can support an input block size exceeding $10^{8}$ bits with an output size that can be flexibly adjusted. To our knowledge, our scheme outperforms existing FPGA-based DV-QKD PA schemes by two orders of magnitude in terms of input block size. Notably, this significant improvement cannot be achieved merely by increasing parallelism.

The rest of this paper is organized as follows: As the foundation of this work, Section \Rmnum{2} first discusses the Universal hash families in PA, followed by a brief introduction to the new proposed universal hash family, and the large-scale PA algorithm. Section \Rmnum{3} describes the design rationale of main hardware modules and our optimizations in implementing the algorithm. In Section \Rmnum{4}, performance results on the selected FPGA device and comparisons with related works are presented. Finally, we draw conclusions for the proposed PA scheme and discuss directions for future work in Section \Rmnum{5}.

\section{Large-scale Privacy Amplification Algorithm}

In this section, we begin by discussing widely used universal hash families \cite{carter1979universal} for PA and analyzing their advantages and disadvantages. Based on this analysis, we introduce a new universal hash family for PA, called the Dynamic Multi-stage Multilinear-Modular Hashing family (DM3H), and prove its universal property. Utilizing the DM3H family, we propose a new PA algorithm suitable for FPGA implementation, which can support larger input block sizes.

\subsection{Existing Universal Hash Families in Privacy Amplification}

A universal hash family ensures that the keys extracted in the PA process are highly secure, such that any eavesdropper (Eve) has only a negligible probability of successfully guessing the secure key. Currently, the two most commonly used universal hash families in PA are Toeplitz matrix and modular arithmetic hash families.

Toeplitz matrix \cite{Krawczyk1994} is the most widely used universal hash family in PA. An $l \times n$ Toeplitz matrix can be defined as follows, it can transform $n$-bits input key into $l$-bits secure key:
\begin{equation}
  T_{l\times n} = \begin{pmatrix}
    t_{n} & t_{n-1} & \cdots & t_{2} & t_{1} \\
    t_{n+1} & t_{n} & \cdots & t_{3} & t_{2} \\
    \vdots & \vdots & \ddots & \vdots & \vdots \\
    t_{n+l-2} & t_{n+l-3} & \cdots & t_{l} & t_{l-1} \\
    t_{n+l-1} & t_{n+l-2} & \cdots & t_{l+1} & t_{l} \\
    \end{pmatrix}_{l\times n}
\end{equation}
The elements in Toeplitz matrix on each diagonal line are constant, therefore, the matrix can be represented by $n+l-1$ bits random numbers.
The Fast Fourier Transform (FFT) and Fast Number Theoretic Transform (NTT) can be used to effectively reduce computational complexity from $O(n^{2})$ to $O(nlogn)$.

Modular arithmetic hash family \cite{carter1979universal} is another critical universal hash family used in PA. The primary operations of Modular arithmetic family involve large integer multiplications and modular operations.
\begin{definition}
  Let $\alpha$ and $\beta$ be two strictly positive integers, $\alpha > \beta$. Define a family modular arithmetic hashing of functions from $2^\alpha$ to $2^\beta$ as follows:
  \begin{equation}
    {\rm{MH}}: = \left\{ {{h_{b,c}}:{Z_{{2^\alpha }}} \to {Z_{{2^\beta }}}\left| {b,c \in {Z_{{2^\alpha }}}} \right.,\gcd (b,2) = 1} \right\}\end{equation}
    where the function $h_{b,c}$ is defined as follows:
    \begin{equation}{h_{b,c}}(x): = {{\left( {b \cdot x + c\bmod {2^\alpha }} \right)} \mathord{\left/
    {\vphantom {{\left( {b \cdot x + c\bmod {2^\alpha }} \right)} {{2^{\alpha  - \beta }}}}} \right.
    \kern-\nulldelimiterspace} {{2^{\alpha  - \beta }}}}
  \end{equation}
\end{definition}
Large integer multiplication can be accelerated using various algorithms, such as the Karatsuba algorithm, the Toom-Cook algorithm, the Schönhage-Strassen algorithm, and the Fürer algorithm.

The advantage of Toeplitz hashing PA lies in its ability to freely split and process the input key sequence. However, FFT-accelerated Toeplitz matrix multiplication involves floating-point operations which demand significant resources and may introduce cumulative errors, potentially affecting accuracy at large scales.

Additionally, the advantage of modular arithmetic hashing in PA is that it allows the input key to be processed as a high-radix integer rather than a binary sequence, effectively reducing the input length in multiplication optimization \cite{Yan2022}. However, the limited resources of FPGA impose constraints on the scale of large integer multiplication.

Combining the advantages of the Toeplitz matrix and modular arithmetic hashing, the MMH was first applied to PA algorithms in \cite{Yan2023}. Specifically, the MMH can split the input key into multiple integers for parallel processing, and reduce computation complexity to $O(nlogn)$ by NTT-based optimization, allowing for efficient secure key distillation.
The MMH family has been proven to be a universal hash family \cite{Halevi1997}, with a collision probability of $1/Z_{p}$.

The definition of the MMH family is given as follows.
This construction operates within the finite field $\mathbb{Z}_p$, where $p$ is a prime number. The family of hash functions comprises all multilinear functions defined over $\mathbb{Z}_p^n$, where $n$ is an integer.
\begin{definition}
  Let $p$ be a prime and let $n$ be an integer $n > 0$. Define a family $\text{MMH}$ of functions from $\mathbb{Z}_p^n$ to $\mathbb{Z}_p$ as follows: 
	\begin{equation}
    \text{MMH} \coloneqq \{h_A: \mathbb{Z}_p^n \rightarrow \mathbb {Z}_p \mid X \in \mathbb{Z}_p^n, A \in \mathbb{Z}_p^n\}
  \end{equation}
\end{definition}

Denote the family of MMH as $H_A$, then the functions $h_A \in H_A$ are defined for any $X = \langle x_1, \ldots, x_n\rangle \in \mathbb{Z}_p^n$, $A = \langle a_1, \ldots, a_n \rangle \in \mathbb{Z}_p^n$ by
\begin{equation}
  h_A(X) \coloneqq \sum_{i=1}^{n} a_i \cdot x_i\mod p
\end{equation}

However, a significant limitation of the MMH family in PA is the constrained output size. Specifically, the maximum output size achievable by MMH is only $1/n$ of the input block size. 

\subsection{New Optimized Universal Hash Family}

In the following, we introduce an optimized universal hash family specifically designed to address the limitation of supporting sufficiently large output lengths. Concatenating multiple smaller sequences to construct a larger output key has been proven to be an effective method \cite{Sciences2006}. Inspired by this approach, we optimized the MMH by dynamically extracting multiple shorter output key sequences from the same original input key using different random sequences. These shorter sequences are then concatenated to form the final output key. It has been demonstrated that DM3H can handle an input block size exceeding $10^8$ bits, while maintaining arbitrary output sizes by flexibly choosing the number of operations $m$. The definition of DM3H is as follows:

\begin{definition}
  Let $p$ be a prime and let $n$ and $m$ be an integer, $n \ge m > 0$. Define the family $\text{DM3H}$ as follows: 
	\begin{equation}
    \text{DM3H} \coloneqq \{h_A : \mathbb{Z}_p^n \rightarrow \mathbb{Z}_p^m \mid X \in \mathbb{Z}_p^n, A \in \mathbb{Z}_p^{n+m-1}\}
  \end{equation}
\end{definition}

The result of the hash function $h_A(X)$ is the concatenation of the results of functions $f_1(X)$ to $f_m(X)$. In this formula, '$\Vert$' represents the concatenation symbol.

\begin{equation}
    \label{eq_def_DIMMH}
  h_A(X) \coloneqq f_1(X) \Vert f_2(X) \Vert \ldots \Vert f_m(X)
\end{equation}

In the equation above, for any $X = \langle x_1, \ldots, x_n\rangle \in \mathbb{Z}_p^n$, and $A = \langle a_1, \ldots, a_n , \ldots, a_{n+m-1} \rangle \in \mathbb{Z}_p^{n+m-1}$, the MMH functions $f_1,\ldots,f_m$ are defined as follows:
\begin{equation}
  f_i(X) \coloneqq \sum_{j=1}^{n} a_{j+i-1} \cdot x_j\mod p
\end{equation}

To incorporate the DM3H family into the PA algorithm, it is necessary to prove its universal property to ensure the security of the PA algorithm. The proof is as follows:

\begin{theorem}
  The DM3H family is Universal.
\end{theorem}

\begin{proof}
  Consider the situation of hash collision, for $X_1,X_2\in {Z}_p^n,X_1 \ne X_2$, the following result holds:
  \begin{equation}
    h(X_1)=h(X_2)=Y,\quad Y\in Z_p^m
  \end{equation}
  According to Equation \ref{eq_def_DIMMH}, the necessary condition for the above equation to hold is that:
  \begin{equation}
  f_i(X_1)=f_i(X_2), \quad \forall i \in \{1, 2, \dots, m\}
  \end{equation}
  Since $f(x)=\sum_{j=1}^{n} a_{j} \cdot x_j\mod p$, the hash function $f$ is additive. Let $X_1-X_2$ be denoted as $X^{'}$, the above formula can be rewritten as:
  \begin{equation}
    f_i(X_1-X_2)=f_i(X^{'})=0, \quad \forall i \in \{1, 2, \dots, m\}
  \end{equation}
  Since $X_1$ is not equal to $X_2$, we have that $X^{'}\ne 0$. 
  For fixed $X^{'}=\langle x_1, \ldots, x_n\rangle \in \mathbb{Z}_p^n$, in order to obtain the number of hash functions $h_A$ that satisfy the condition above, we can rewrite the formula as a system of linear equations.
  \begin{equation}
    \resizebox{.9\hsize}{!}{$%
      \begin{pmatrix}
        x_1 & x_2 & x_3 & \cdots &x_n & 0 & \cdots & 0 & 0 \\
        0 & x_1 & x_2 & \cdots &x_{n-1} & x_n & \cdots & 0 & 0 \\
        \vdots & \vdots & \ddots & \ddots & \ddots & \ddots & \ddots & \vdots & \vdots\\
        0 & 0 & \cdots & x_1 & x_2 & \cdots & \cdots & x_n & 0\\
        0 & 0 & \cdots &  & x_1 & x_2 & \cdots & x_{n-1} & x_n
      \end{pmatrix}_{m \times \left(n+m-1\right)}
      \cdot
      \begin{pmatrix}
        a_1 \\ a_2 \\ \vdots \\ a_{n+m-2} \\ a_{n+m-1}
      \end{pmatrix}
      =0
    $}
  \end{equation} 
  Since $X^{'}\ne 0$, the rank of the coefficient matrix of this system of linear equations is $m$, which restricts $m$ variables, leaving $n-1$ variables free. In the finite field, each variable can take $p$ distinct values, the total number of solutions to this system of equations is determined by the number of possible combinations of values for the free variables, which is $p^{n-1}$. So we have following result:
  \begin{equation}
      \left| \left\{ h_A\ :\ h\left( X_1 \right) =h\left( X_2 \right) |X_1\ne X_2 \right\} \right| \le p^{n-1}
      =\frac{\left| H \right|}{\left| \mathbb{Z}_p^m \right|}
  \end{equation}
  According to the definition of the Universal hash function family, the DM3H family is Universal.
\end{proof}

The DM3H family provides an efficient method for handling large input block sizes while significantly reducing the consumption of random numbers, which are crucial for DV-QKD systems. Leveraging the structure of the DM3H family, input key sequences can be efficiently divided into multiple sub-blocks, making it particularly suitable for FPGA-based PA algorithms.

\subsection{Novel Hybrid Privacy Amplification Algorithm}

Theoretically, the DM3H is a universal hash family suitable for PA. However, in practical QKD systems, the input key sequence for the PA process typically consists of a binary sequence rather than multiple large prime numbers. To address this discrepancy, we have configured the prime p within the DM3H as a Mersenne prime, denoted as $p = 2^{\gamma}-1$. This configuration facilitates a more effective translation of the PA inputs, which are elements from the set $\{0,1\}^{\gamma \times n}$, to the finite field $\mathbb{Z}_p^{n}$, thereby enhancing the efficiency and applicability of the PA process in high-performance QKD systems. The only exception occurs when a segment of the input sequence, denoted as $x_i$, equals $2^{\gamma}-1$. In such cases, this segment of the input key should be discarded and replaced. However, the probability of this occurrence is $1/{2^{\gamma}}$, which becomes negligible as $\gamma$ increases.

Nevertheless, due to the structural characteristics of the DM3H, its output size must be a multiple of $\gamma$, where $\gamma$ approaches $10^6$ bits. This significantly restricts the flexibility of PA implementation. To address this limitation, we concatenate the output of the DM3H with that of another universal hash function to achieve an adaptive output size. This approach also meets the necessary security requirements for PA. For the convenience of hardware implementation, we employ a modular arithmetic hash family, Multilinear Modular Hashing-Modular Arithmetic Hashing (MMH-MH) family, which can generate an output key of reduced but flexible length from the input. Additionally, its primary computation involves large integer multiplication, allowing for the reuse of the multiplication modules as DM3H, thereby reducing design complexity.

The definition of MMH-MH family are given as follows:
\begin{definition}
  If we denote the MMH family as $H_A$, and the MH family as $H_{B,C}$, then we can denote the MMH-MH family as $H_2={H_{B,C}}(H_A)$, where the function $h_2 \in H_2$ is defined as follows:
  \begin{equation}
    h_2(x):= {h_{b,c}}(h_a(x))
  \end{equation}
  Here, $h_a \in H_A$ and $h_{b,c}\in H_{B,C}$ are hash functions from their respective families.
\end{definition}

The main steps of the hybrid PA algorithm are illustrated as  Algorithm \ref{algo1}.

\begin{algorithm}
  \caption{Novel Hybrid PA Algorithm}\label{algo1}
  \begin{algorithmic}[1] 
  \Statex \textbf{Input:}
  \Statex Input key sequence: $X \in \{0,1\}^N$.
  \Statex Random Seeds: $A \in \mathbb{Z}_{p}^{n+m-1}$, $b \in \mathbb{Z}_{2^{\gamma}}$, $c \in \mathbb{Z}_{2^{\gamma}}$.
  \Statex // $p = 2^{\gamma} - 1$, $\gcd(b, 2) = 1$
  \Statex \textbf{Output:}
  \Statex Output key sequence $K \in \{0,1\}^l$, $l=m\times\gamma+l'$

  \State $X' = (X, \text{padding}\; 0)$, $X' \in \{0,1\}^{\gamma \times n}$// Padding
  
  \State $X' = (x_1, \ldots, x_n)$ // Split data $X$
  \State $A = (a_1, \ldots, a_n, \ldots, a_{n+m-1})$ // Split data $A$
  \If{$x_i = 2^{\gamma} - 1$ for any $i = 1, \ldots, k$}
      \State  Reload data $x_i$
  \EndIf

  \Statex// DM3H function
  \For{$i = 1$ to $m+1$}
    \For{$j = 1$ to $n$}
        \State $y''_j = \text{NTT}(x_j) \cdot \text{NTT}(a_{j+i-1})$
        \State $y'_j = \text{INTT}(y''_i)$
    \EndFor
    \State $y_i = \sum_{j=1}^n y'_j \mod p$ 
  \EndFor

  \Statex// MH function
  \State $z'' = \text{NTT}(y_{m+1}) \cdot \text{NTT}(b)$
  \State $z' = \text{INTT}(z'')$
  \State $z = (z' + c \mod 2^{\gamma})^{2^{l'}}$
  
  \State \textbf{return} $K=(y_1, \ldots, y_m, z)$
  \end{algorithmic}
\end{algorithm}

First, the output secure key length, denoted as $l$, is calculated based on the parameters of the QKD system. Using the value of $l$ and the preset block size $\gamma$, the number of output blocks $m$ and the length of the remaining part $l'$ can be determined, where $l=m\times\gamma+l'$. To align with the constructions of DM3H and MMH-MH, the length of the input key $X$, denoted as $N$, must be padded to a multiple of $\gamma$. This padded input sequence is denoted as $X^{'}$. 

The distillation process employs both the DM3H and MMH-MH hash families. The padded input sequence $X^{'}$ is processed using the DM3H function $h_1 \in H_1$, resulting in $m$ blocks, each of length $\gamma$. Concurrently, the remaining part of the output key, with length $l'$, is generated using the MMH-MH function $h_2 \in H_2$. Finally, the secure key of length $l$ is obtained by concatenating the outputs from the DM3H and MMH-MH functions. The security of the final key is rigorously proven, ensuring its robustness against potential threats \cite{Sciences2006}.


\section{Implementation of large-scale PA algorithm}

In this section, we present the implementation of the large-scale PA algorithm on an FPGA platform. While the FPGA platform offers high parallelism, supporting input sizes exceeding $10^8$ bits remains a significant challenge due to limited hardware resources.

Therefore, we analyzed the computational complexity of the PA algorithm and chose parameters that balance performance and resource constraints. A dual-path parallel NTT core was designed to improve the computational efficiency of large integer multiplication. Additionally, a high-speed modular accumulation unit was designed to support efficient pipelining for large integer accumulation. By reusing the large integer multiplication module and employing efficient pipelining, we achieved high throughput when processing large-scale input and output.

\begin{figure}[!t]
  \centering
  \includegraphics[width=3.5in]{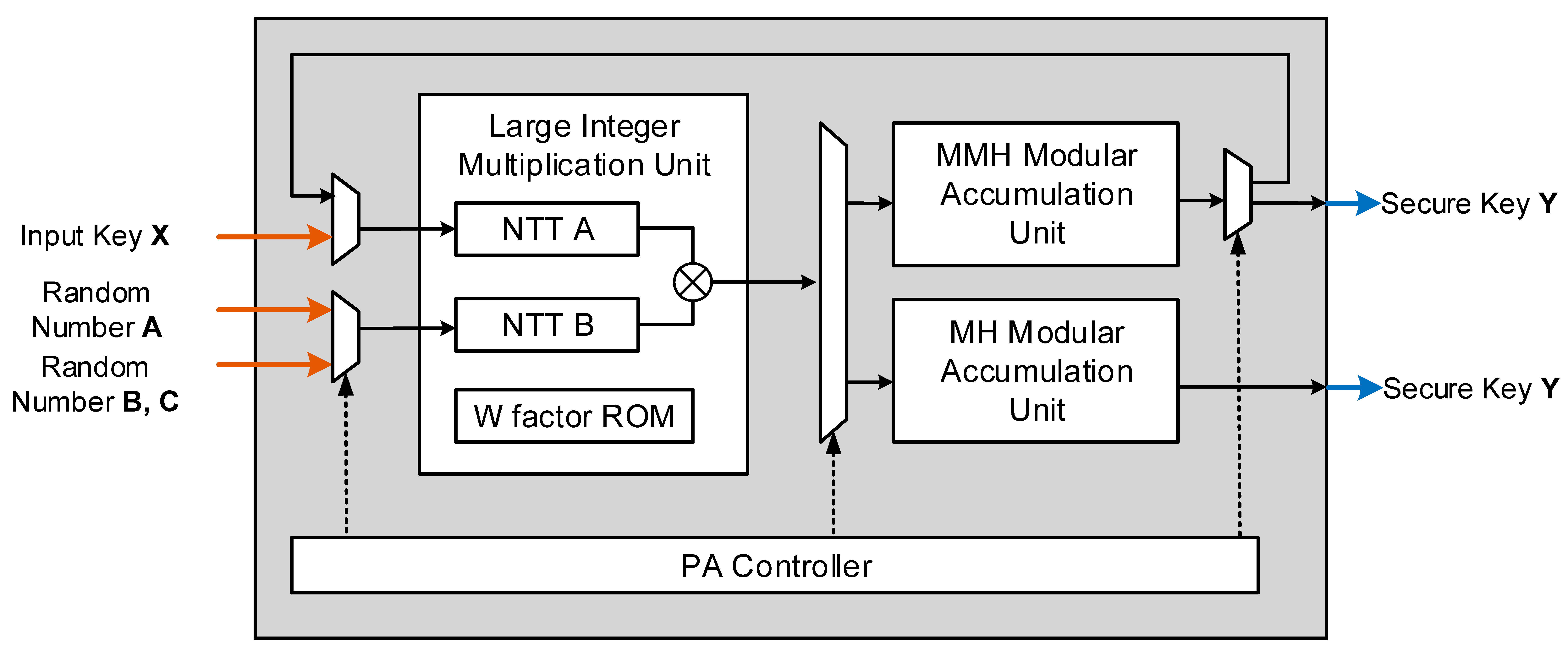}
  \caption{Schematic of the FPGA implementation of hybrid PA algorithm. Solid lines represent data signals, while dashed lines represent control signals.}\label{fig_archi}
\end{figure}

\subsection{Top-Level Architecture Design}

The top-level architecture of the PA implementation is shown in Figure \ref{fig_archi}. The implementation comprises several key modules: large integer multiplication unit, MMH modular accumulation unit, and MH modular accumulation unit. The large integer multiplication unit is designed to efficiently handle large number multiplications, which are essential for preprocessing input keys and random numbers. The MMH modular adder module accumulates results from the MMH function, that is $\sum_{j=1}^{n} a_{j+i-1} \cdot x_j\mod p$ in DM3H, while the MH modular adder module computes the results of the modular arithmetic hash function, that is $\left(b \cdot x + c \right) \bmod {2^\alpha }$ in MH.

The hybrid PA algorithm utilizes the DM3H and MMH-MH families to distill the input key. When the number of blocks output by DM3H is $m$, the MMH computation needs to be performed $m+1$ times. Each of these $m+1$ computations is performed independently. The MMH modular adder module accumulates the results from the MMH function. The results of the first $m$ computations are directly output as part of the final key, while the result of the last computation undergoes additional compression through modular arithmetic hashing.

In this process, a crucial parameter is the sub-block length $\gamma$, where $2^{\gamma}-1$ is constrained to be a Mersenne prime, thus limiting the choice to the following values:
\begin{equation}
\begin{aligned}
\gamma \in \{ & 110503, 132049, 216091, 756839, 859433,\\
& 1257787, 1398269, 2976221, 3021377, 6972593,\\ 
& 13466917, 20996011, 24036583, 25964951,\\ 
& 30402457, 32582657, 37156667, 42643801,\\
& 43112609, 57885161, 74207281 \}
\end{aligned}
\end{equation}

Let $B$ denote the sub-block length. The computational complexity of this scheme can be approximately evaluated as $C \cdot \left\lceil \frac{n}{B} \right\rceil \cdot \left\lceil \frac{l}{B} \right\rceil \cdot B\log B$, where $n$ and $l$ represent the input block size and the output size, respectively, and $C$ is a constant.

Theoretically, increasing the sub-block length reduces the computational complexity, thereby enhancing the overall efficiency of distillation. However, due to limitations in FPGA resources, the scale of large integer multiplication cannot be increased indefinitely. Ultimately, we decide to set the block length $\gamma$ to 756839, which represents a balanced choice.

\begin{figure}[!t]
  \centering
  \includegraphics[width=3.5in]{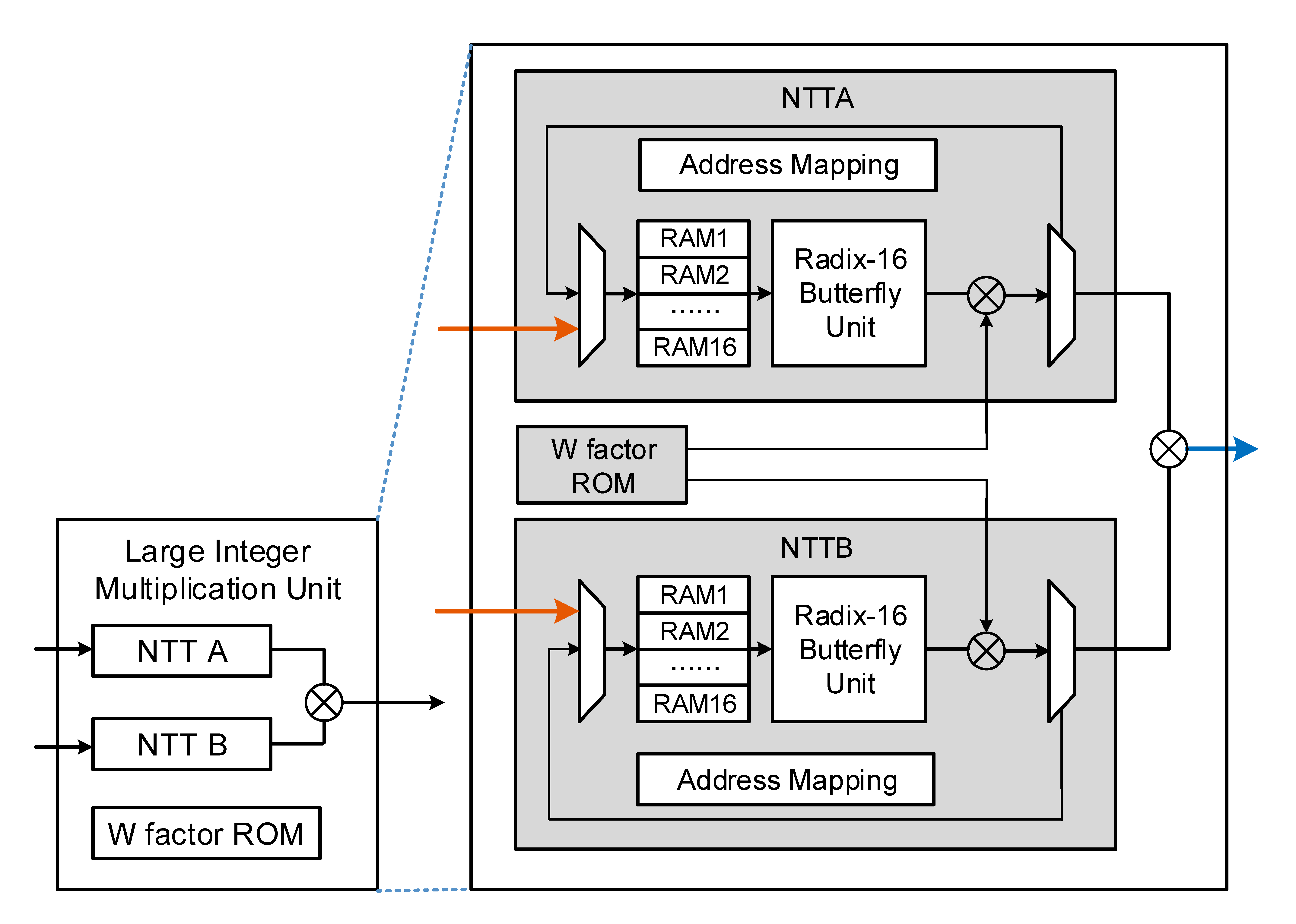}
  \caption{Structure of large integer multiplication unit.}\label{fig_mul}
\end{figure}

\subsection{Large Integer Multiplication Unit}

We utilize NTT acceleration to reduce the complexity of block multiplication from $O(n^2)$ to $O(n\log n)$. A 65,535-point NTT core is implemented to support the block-wise multiplication. To optimize the resource consumption of the NTT core, we set the data width of each point to 24 bits. Consequently, each NTT core can support a maximum input length of $32768 \times 24 = 786432$ bits. Therefore, we selected a block length $\gamma$ of 756,864, as $2^{756864}-1$ is a Mersenne prime, and this length closely matches the multiplier length supported by the NTT core.

As shown in Figure \ref{fig_mul}, the large integer multiplication unit consists of two parallel NTT cores and a multiplier, which can perform modular multiplication with a maximum length of 786,432 bits.

The principle of NTT-accelerated large integer multiplication is illustrated as follows. For N-point NTT, it divides a large integer $x$ into a sequence $x=(x_0, x_1,\ldots, x_{N-1})$, and then transforms the sequence using:
\begin{equation}
  X_k=\sum_{n=0}^{N-1}{x_n}\cdot w_{N}^{nk}\mod\ p
\end{equation}
where $p$ is a prime, and $w_N$ is a primitive n-th root of modulo $p$. The inverse transform INTT is defined as follows:
\begin{equation}
  x_n=\frac{1}{N}\sum_{k=0}^{N-1}{X_k}\cdot w_{N}^{-nk}\mod\ p
\end{equation}

The NTT can accelerate the multiplication of large integers by transforming the normal multiplication operation into simpler pointwise multiplications. In this NTT implementation, we chose a special prime $p=2^{64} - 2^{32} + 1$ as the modulus, where the 16-th root of unity $w_{16}=4096=2^{12}$ is a power of 2. This allows the multiplication by the root of unity $w_{16}^{n}$ during the butterfly operation to be replaced by a simple bit-shift operation, as multiplying by $w_{16}^{n}$ is equivalent to a left shift by $12 \times n$ bits.

This NTT implementation utilizes a radix-16 iterative approach instead of the traditional radix-2 method, allowing a 65,536-point NTT computation to be completed in 4 iterations. This approach effectively reduces the number of iterations, thereby enhancing the efficiency of the NTT.

\subsection{Modular Accumulation Unit}

The accumulate operations in MMH and MH are quite similar. Therefore, a typical modular accumulation unit can be used to describe them, as shown in Figure \ref{fig_acc}. The modular accumulation process has been optimized as follows:

\begin{equation}
  \begin{aligned}
    &(a + b) \mod \left(2^{756839} - 1\right)\\ 
    &= a \mod 2^{756839} + b \mod 2^{756839} \\
    & + \left\lfloor \frac{a}{2^{756839}} \right\rfloor + \left\lfloor \frac{b}{2^{756839}} \right\rfloor
  \end{aligned}
\end{equation}

From the equation above, it is clear that the modular accumulation process requires additional processing for data exceeding 756,839 bits at the end of the operation. All of these accumulation operations are performed using the same set of full adder arrays and storage modules. While this approach demands higher control precision, it significantly reduces the consumption of computational and storage resources.

\begin{figure}[!t]
  \centering
  \includegraphics[width=3.5in]{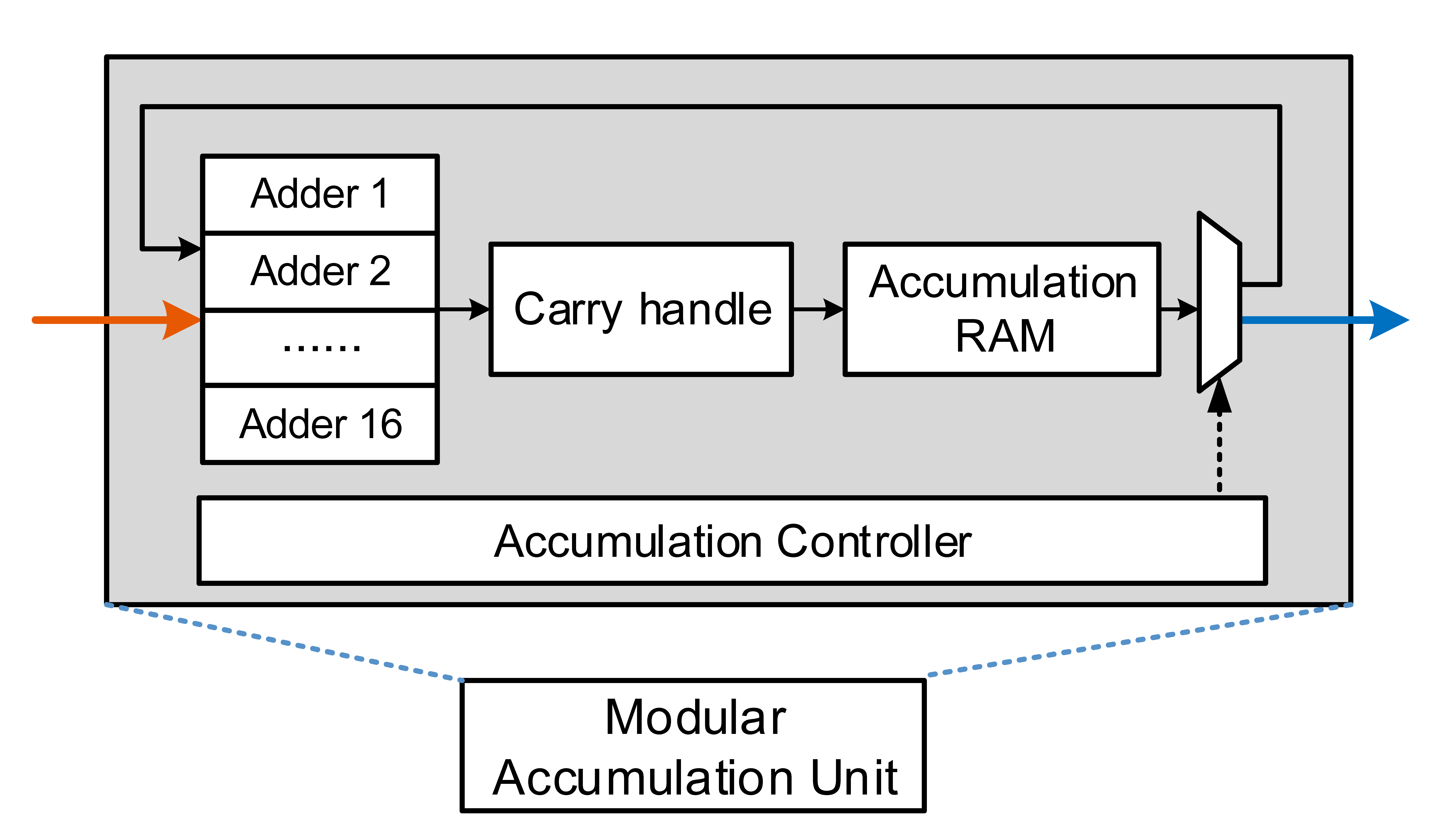}
  \caption{Optimization of binary modular and division operations.}\label{fig_acc}
\end{figure}

\section{Experiments and performance analysis}

To validate the performance of our proposed PA scheme, we implement our PA scheme on the Xilinx XCKU095 FPGA. In the following parts, we summarize the implementation results of the proposed PA scheme and compare them with the previous PA schemes in several aspects. The performance of our PA scheme is analyzed in several aspects.

\subsection{Input Block Size and Output Size}

Large input block sizes can effectively mitigate the finite-size effect, significantly enhancing the secure key rate of QKD systems. According to the experimental results on the T12 protocol by Lucamarini et al. \cite{Lucamarini2013}, an input block size of $10^8$ bits effectively mitigates the finite-size effect, allowing the secure key rate to reach 85\% of its asymptotic value on a fiber distance of 50Km. While the secure key rate falls to less than 50\% of the asymptotic value with an input block size of less than $10^7$ bits. This significant difference highlights the importance of sufficiently large input block sizes in achieving optimal performance in QKD systems.

In order to effectively compare the capability of different PA schemes in supporting output sizes, we introduce the concept of compression ratio, defined as the ratio between the input block size and the output size, as a convenient metric for evaluation.
Theoretically, our scheme can accommodate input block sizes of arbitrary length. However, in practical implementation, given a fixed compression ratio, an excessively large input block size results in an increased secure key length. Since the block length of our implemented large integer multiplication is fixed, the number of output blocks increases, thereby increasing the number of PA iterations required. Consequently, this affects the overall throughput of the PA scheme.

To demonstrate the practicality of the proposed PA scheme, we compared the input block size and throughput of our PA scheme with those of existing FPGA-based PA schemes, as shown in Figure \ref{fig_7}. The figure also illustrates the relationship between input block size and throughput for the proposed scheme under different compression ratios.

\begin{figure}[!t]
  \centering
  \includegraphics[width=3.5in]{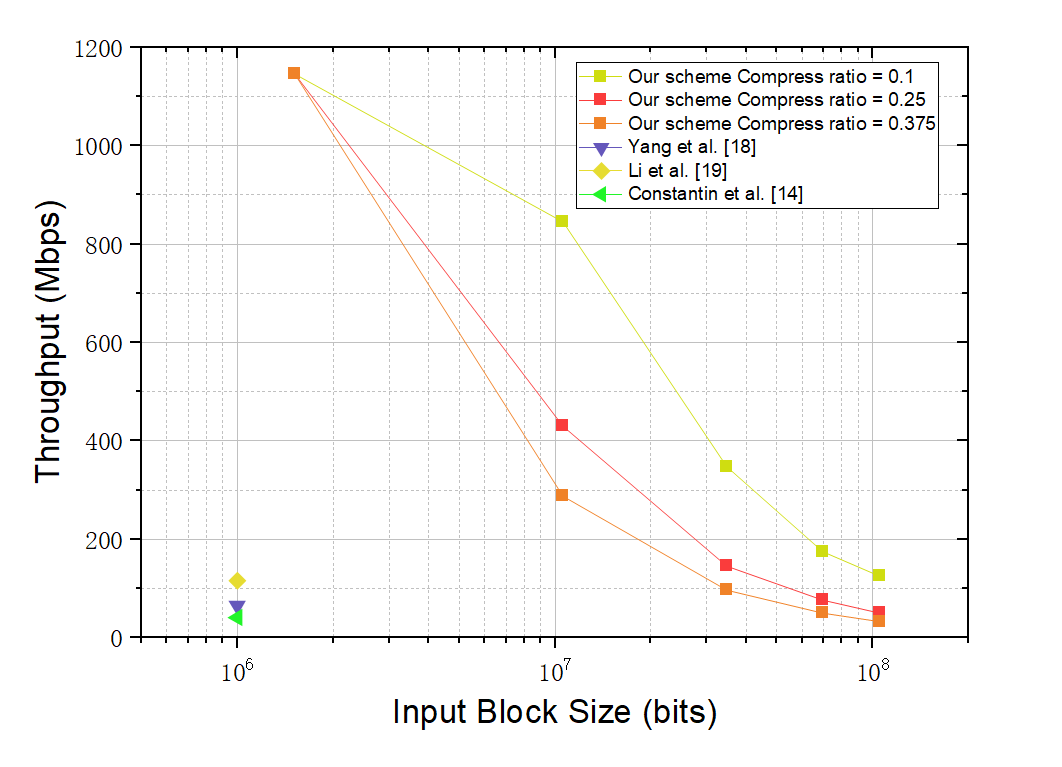}
  \caption{Comparison of input block size and throughput in FPGA-based PA schemes.}\label{fig_7}
\end{figure}

The comparison indicates that our proposed scheme can support input block sizes that are two orders of magnitude larger than existing schemes at the same compression ratio while maintaining throughput comparable to existing schemes.

The proposed PA scheme demonstrates high performance in handling large-scale inputs. Additionally, its ability to flexibly adjust input and output sizes makes it particularly well-suited for high-performance DV-QKD systems, ensuring flexibility and efficiency in various operational scenarios.

\subsection{Resource Consumption}

Table \ref{tab_res} provides a comparative analysis of resource consumption between our PA implementation and existing FPGA-based PA schemes. Resource costs include look-up tables (LUTs) and Block RAM (BRAM), while the key metrics are input block size and compression ratio. To demonstrate the practical performance of the scheme, throughput metrics are also incorporated into the comparison. In this table, throughput is presented under a compression ratio of 0.1 to ensure consistency.

\begin{table}
  \begin{center}
  \caption{The performance and resource consumption comparison of FPGA-based PA
 schemes}
  \label{tab_res}
  \begin{tabular}{ccccc}
  \hline
   & \cite{Yang2017} & \cite{Constantin2017} & \cite{Li2019} & Our proposed \\
  \hline
  FPGA & Virtex 7 & Virtex 6 & Virtex 6 & \makecell[c]{Kintex\\Ultrascale} \\
  LUTs & 15604 & 26571 & 37203 & 155757\\
  BRAM & 100kb & 0kb & 5.5Mb & 18.8Mb\\
  Max Block Size  & $10^6$ bits & $10^6$ bits & $10^6$ bits & \textbf{$\bm{10^8}$ bits}\\ 
  Compress Ratio  & (0, 1) & (0, 1) & (0, 1) & (0, 1)\\
  Throughput & 65.4Mbps & 41Mbps & 116Mbps & 125Mbps \\
  \hline 
  \end{tabular}
  \end{center}
\end{table}

The primary resource consumption in the proposed PA scheme stems from the NTT-based large integer multiplication units. As the input block size increases, the required number of NTT iterations also increases. To minimize the computational cost of the PA, the multiplication units must support larger multiplication numbers. In our PA implementation, the NTT supports up to 65,536 points and a maximum coefficient size of 64 bits, resulting in a per-round multiplication block size close to $10^6$ bits.

Although our scheme requires more resources, this investment is worthwhile. Specifically, our scheme achieves a more than 100-fold increase in input block size, which nearly corresponds to a 100-fold increase in computational complexity, with less than five times the resource consumption, significantly mitigating the impact of the finite-size effect on the secure key rate of DV-QKD systems. It is important to note that simply increasing resource usage through parallelism does not result in larger input block sizes, because the output key in privacy amplification must correlate with every bit of the input key. As such, the entire input key cannot be split into multiple sequences for processing but must be processed as a single, unified entity during the hashing operation. Our approach leverages architectural and algorithmic optimizations to achieve this significant improvement. Furthermore, resource consumption remains constant even as input block sizes and compression ratios increase, making it advantageous for QKD systems handling large input sizes.

This is the first FPGA-based PA scheme that can support both large input block sizes and high compression ratios, which is highly significant for enhancing the performance and security of DV-QKD systems.

\section{Conclusion}

In this paper, we propose a novel hybrid PA algorithm that supports input block sizes exceeding $10^8$ bits at arbitrary compression ratios. Furthermore, we implement the PA scheme on a Kintex UltraScale FPGA and demonstrate its high performance. Experimental results show that the proposed PA scheme can handle input block sizes exceeding $10^8$ bits at typical DV-QKD compression ratios. To our knowledge, the achieved input block size is two orders of magnitude larger than that of existing FPGA-based PA schemes for DV-QKD, which is highly significant for mitigating the finite-size effect and enhancing the secure key rate of QKD systems. These results indicate that the proposed scheme meets the requirements of current high-performance DV-QKD systems, offering a new solution for FPGA-based PA in QKD applications.

Given the resource limitations of the FPGA, when the input block size of the PA scheme reaches $10^9$ bits, the throughput is limited to only 12 Mbps, which becomes a bottleneck for practical QKD systems. In future work, we will continue to optimize the NTT unit. These optimizations are expected to enable the PA scheme to support even large input block sizes with high compression ratios, thus making it more suitable for high-performance DV-QKD systems in practical applications.

\section{Citations to the Bibliography}


\end{document}